\documentclass[11pt]{article}

\usepackage[english]{babel}

\usepackage[letterpaper,top=2cm,bottom=2cm,left=3cm,right=3cm,marginparwidth=1.75cm]{geometry}

\usepackage{amsmath}
\usepackage{graphicx}
\usepackage[colorlinks=true, allcolors=blue]{hyperref}
\usepackage{setspace}
\usepackage{natbib}
\usepackage{booktabs}
\usepackage{multirow}
\usepackage{array}
\usepackage{amsthm}
\usepackage{amssymb}
\usepackage{caption} 
\usepackage[perpage]{footmisc}
\setcitestyle{authoryear,round}
\newcolumntype{L}[1]{>{\raggedright\arraybackslash}p{#1}}
\newcolumntype{C}[1]{>{\centering\arraybackslash}p{#1}}
\newtheorem{theorem}{Theorem}

\title{\textbf{Heterogeneous Agents in the Data Economy}}
\author{Yongheng Hu\footnote{School of International Business, Zhejiang International Studies University, Liuhe Road, Hangzhou 310023, China. Correspondence to: Yongheng Hu (22030101043@st.zisu.edu.cn). This working paper is incomplete and still being improved, all comments and opinions about the article are welcome. Of course all remaining omissions and errors in statement or technique are mine.}}

\begin{document}
\maketitle

\begin{abstract}
In this short paper, we define the investment ability of data investors in the data economy and its heterogeneity. We further construct an analytical heterogeneous agent model to demonstrate that differences in data investment ability lead to divergent economic results for data investors. The analytical results prove that: Investors with higher data investment ability can obtain greater utility through data investment, and thus have stronger incentives to invest in a larger scale of data to achieve higher productivity, technological progress, and experience lower financial frictions. We aim to propose a prerequisite theory that extends the analytical framework of the data economy from the currently prevalent representative agent model to a heterogeneous agent model.

\textbf{Key Words}: Data Economy, Heterogeneous Agents, Data Investment Ability, Data Cost, Financial Friction

\textbf{JEL Codes}: B22, E22
\end{abstract}

\newpage

\section{Introduction}

The rise of data as an economic factor fundamentally shifts the material foundation for information and knowledge production from traditional discursive practices to quantifiable, tradable data resources. This process completes the dataization of information and knowledge and empowers technological innovation (\citealp{costa2024socio}; \citealp{costa2025exploring}; \citealp{gans2025a}; \citealp{gans2025b}), then, it will drive the intelligent transformation of industries, optimize the allocation of production resources, reduce production costs, and enhance production efficiency (\citealp{einav2014economics}; \citealp{schaefer2014long}; \citealp{zhang2018survey}; \citealp{glaeser2018big}; \citealp{corrado2024data}; \citealp{abis2024changing}; \citealp{athey2025presidential}; \citealp{acemoglu2025simple}). Data possesses unique economic attributes as a new kind of production factor, such as non-rivalry, replicability, virtuality, open-source nature, and physical non-destruction (\citealp{lynch2008your}; \citealp{goldfarb2019digital}; \citealp{jones2020nonrivalry}), enabling its integration into the capital accumulation cycle with diminishing marginal costs and increasing returns to scale. This enhances the marginal productivity of traditional production factors while broadening information dissemination channels, decreasing barriers to acquiring knowledge and skills, which substantially makes total factor productivity and innovation levels prosperous (\citealp{graetz2018robots}; \citealp{acemoglu2018race}; \citealp{bessen2019automation}; \citealp{farboodi2019big}; \citealp{farboodi2021model}; \citealp{hemous2022rise}; \citealp{nose2023inclusive}). Meanwhile, this has gradually impacted the quality of human capital (enhancing or diminishing it) and the demand structure for the social labor resource (updating labor skills) (\citealp{aum2025labor}; \citealp{beraja2025inefficient}; \citealp{brynjolfsson2025generative}; \citealp{hampole2025artificial}; \citealp{jiang2025ai}; \citealp{weidmann2025measuring}).

However, is the economic prosperity generated by data elements attributable to all agents in society possessing homogeneous data utilization abilities? Or is it because individual agents with higher levels of data utilization abilities are able to leverage data elements to empower economic development better ? This involves the application of representative agent and heterogeneous agent models in the data economy. Most existing literature on the data economy assumes agents possess representative data utilization abilities (\citealp{jones2020nonrivalry}; \citealp{hu2025analysis}). Although this approach ensures model's simplicity, efficiency, and interpretable economic outcomes, it fails to provide a more objective description of the realities within the data economy. We argue that agents in the data economy exhibit varying abilities to utilize data as a factor of production. These disparities in data utilization abilities (the inequality of abilities) lead to distinct economic consequences for data users, including differences in the scale of data investment, output levels, technological progress efficiency, and financial frictions. Therefore, this paper aims to construct a heterogeneous agent model to identify the heterogeneity of economic variables in the data economy caused by differences in agents' data utilization abilities. This serves as a preliminary theoretical foundation for further expanding the analysis of heterogeneous agent models in the data economy.

Current research on the data economy primarily focuses on two distinct directions. The first direction treats data elements (big data) as an innovative medium for disseminating information and knowledge. Data accelerates knowledge diffusion and reduces production uncertainty (\citealp{farboodi2019big}; \citealp{farboodi2021model}), and can influence agents' rationality, thereby driving shifts in their decision-making and indirectly affecting their utility or welfare acquisition (\citealp{hayek1945use}; \citealp{angeletos2018quantifying}; \citealp{caplin2022rationally}; \citealp{hu2025big}; \citealp{ding2025consumer}; \citealp{caplin2025data}; \citealp{bhandari2025survey}; \citealp{angeletos2025inattentive}; \citealp{acemoglu2025big}). This category of research also focuses on the trade-off between the economic benefits that the data users obtain from accessing data and the protection of the data providers’ rights (\citealp{lazer2014parable}; \citealp{cong2021knowledge}; \citealp{cong2023data}; \citealp{acemoglu2024model}; \citealp{jones2024ai}; \citealp{jones2025much}; \citealp{madsen2025insider}). The second direction of research incorporates data as a new type of production factor into endogenous economic models to explore themes such as economic growth, technological progress, optimal allocation of production factors, and industrial structure transformation (\citealp{jones2020nonrivalry}; \citealp{acemoglu2025simple}; \citealp{hu2025analysis}). And the topic in this direction is a major field in the data economy. The data factor discussed in this paper refers to the perspective of \citet{jones2020nonrivalry}, treating data as a factor of production rather than a medium of information or knowledge. Therefore, the analytical framework of our paper is based on the perspective of “data economy”, which is different from digital economy. Data economy is a form of economic research centered on data as a factor of production, which centers on value creation through the production, distribution, exchange and consumption of data resources. Data economy emphasizes that data, as a non-rivalry, reusable resource, can be embedded in the production function to optimize the efficiency of resource allocation and enhance productivity. Its theoretical roots can be traced back to the breakthrough of the assumption of exogenous technology in neoclassical growth theory (\citealp{swan1956economic}; \citealp{solow1957technical}) and the extension of knowledge spillovers in endogenous growth theory (\citealp{romer1986increasing}; \citealp{lucas1988mechanics}; \citealp{romer1990endogenous}). Digital economy refers to the new economic system reconstructed on the basis of Digital Technologies (like AI and Digital Platform), covering the two dimensions of digital infrastructure, digital industrialization and industrial digitization. Its substance is the use of digital general-purpose technology, triggered by the total factor productivity increase and organizational model change, manifested in the expansion of the production boundary and the reduction of transaction costs.

Heterogeneous agent models are currently widely applied in research on topics such as household income, public tax, firm wealth distribution and monetary policy, etc. The application of this model in the data economy is relatively rare. For specific references on the theory and applications of heterogeneous agent models, see: \cite{huggett1993risk}, \cite{aiyagari1994uninsured}, \cite{krusell1998income}, \cite{adrian2010liquidity}, \cite{brunnermeier2014macroeconomic}, \cite{gabaix2016dynamics}, \cite{kaplan2018monetary}, \cite{ahn2018inequality}, \cite{achdou2020mean}, \cite{achdou2022income}, \cite{bilal2023solving}, \cite{fernandez2023financial}, \cite{moll2025mean}, \cite{maxted2025present}, \cite{le2025optimal}. In this paper, we want to avoid complex analysis or calculations and obtain a clean foundational conclusion, hence, we categorize agents in the data economy into two simple types: Data users and data providers. Heterogeneity stems from data users' varying ability to utilize data. This variation in data utilization ability leads to differences in economic variables among heterogeneous agents. Our discussion primarily analyzes how changes in data utilization ability affect data users' scale of data investment, productivity, technological progress, and financial frictions. The key focus is examining how financial frictions faced by data users evolve as their data utilization ability increases.

\section{The Model}
\subsection{Utility Acquisition}

As data users, enterprises accumulate data capital by purchasing data provided by data providers (other enterprises or consumers) and utilize this data capital for investment to obtain utility. Data providers acquire utility by receiving payments from data users in exchange for providing data. Assuming all agents share similar preferences for terminal consumption $U_t$:
\[U_t(C_{i,t+1})=\frac{(C_{i,t+1})^{1-\gamma}}{1-\gamma}\]

Where $C_i(t+1)$ represents the consumption level of agent $i$ in period $t$, and $\gamma$ denotes the risk aversion coefficient. When $\gamma>1$, it corresponds to CRRA utility, when $\gamma=1$, it corresponds to Log utility. We define the agent's data investment ability as the process of converting purchased data from data provider into data capital and utilizing this data capital for investment. Agent $i$ is thus endowed with data investment ability $\mu_i$, which follows the following normal distribution:
\[\mu_i{\sim}\mathcal{N}(0,\sigma_\mu^2)\]

The agent obtains the output $y_{i,t}$ through data investment:
\[y_{i,t+1}=e^{\mu_i+\varepsilon_{t+1}+\varepsilon_{i,t+1}}D_t\]

For each agent $i$ of the data user, all shocks are independent and identically distributed. Therefore, $\varepsilon_{t+1}\sim\mathcal{N}(-\frac{1}{2}\sigma^2,\sigma^2)$ represents the aggregate shock, while $\varepsilon_{i,t+1}\sim\mathcal{N}(-\frac{1}{2}\sigma^2_1,\sigma^2_1)$ denotes the heterogeneous shock, with $\mathbb{E}[e^{\varepsilon_{t+1}}]=\mathbb{E}[e^{\varepsilon_{i,t+1}}]=1$. And $D_t$ represents the direct contribution of data capital to output, embodying a set of factors driven by data elements that enhance output, such as knowledge innovation diffusion, technological progress, capital appreciation, and human capital enhancement. Each agent possesses assets $y_{i,t+1}(1-\tau_t)$ at the investment initiation stage, where $\tau_t$ is the data cost rate, i.e., the expense incurred for data acquisition paid by data users (enterprises) to data providers (consumers or other enterprises). Considering investment risks, each agent must retain ownership of at least a small portion $\vartheta$ of their assets. This friction results in market incompleteness.

We refer to the analytical framework and methodology of \citet{pastor2020political}: Assuming that the total data costs incurred by data users are redistributed proportionally to data providers. In the macroeconomic model of this paper, only data users (enterprises) and data providers (consumers or other enterprises) exist. Let the set of agents i be denoted as $\mathcal{L}_t$. Then the scale of data users is $m_t=\int _{i\in\mathcal{L}_t}di$, and consequently the scale of data providers is $1-m_t$. Given $\mathcal{L}_t$, the economy's expected total output is fixed. A portion of this output is allocated to data providers (share equal to $\tau_t$), while another portion is allocated to data users (share equal to $1-\tau_t$). The consumption of data providers depends on total output $y_t$:
\[y_{t+1}=\int_{j\in \mathcal{L}_t}y_{j,t+1}\]

For a certain data cost rate $\tau_t=\tau$, the total cost is $\tau y_t$:
\[\tau\int_{j\in \mathcal{L}_t}y_{j,t+1}=\tau\left(\int_{j\in \mathcal{L}_t}e^{\mu_j+\varepsilon_{t+1}+\varepsilon_{j,t+1}}dj\right)D_t\]

Based on the \textit{Law of Large Numbers}:
\[\int_{j\in \mathcal{L}_t}e^{\mu_j+\varepsilon_{j,t+1}}dj=m_t\mathbb{E}[e^{\mu_j+\varepsilon_{j,t+1}}|j\in \mathcal{L}_t]=m_t\mathbb{E}[e^{\mu_j}|j\in \mathcal{L}_t]\mathbb{E}[e^{\varepsilon_{j,t}}|j\in \mathcal{L}_t]=m_t\mathbb{E}[e^{\mu_j}|j\in \mathcal{L}_t]\]

Hence:
\[\tau\int_{j\in \mathcal{L}_t}y_{j,t+1}dj=\tau\left(\int_{j\in \mathcal{L}_t}e^{\mu_j+\varepsilon_{t+1}+\varepsilon_{j,t+1}}dj\right)D_t=\tau D_te^{\varepsilon_{t+1}}m_t\mathbb{E}[e^{\mu_j}|j\in \mathcal{L}_t]\]

The total data cost will be evenly distributed among data providers of size $1-m_t$:
\[C_{s,t+1}=\frac{\tau D_te^{\varepsilon_{t+1}}m_t\mathbb{E}[e^{\mu_j}|j\in \mathcal{L}_t]}{1-m_t}\]

The utility of the data provider at stage $t$ is:
\[\begin{gathered}\mathbb{E}_t[U(C_{s,t+1},\gamma\neq1)|\tau]=\frac{\tau^{1-\gamma}}{1-\gamma}D_t^{1-\gamma}\mathbb{E}_t[e^{(1-\gamma)\varepsilon_{t+1}}]\mathbb{E}_t[e^{\mu_j}|j\in\mathcal{L}_t]^{1-\gamma}\left(\frac{m_t}{1-m_t}\right)^{1-\gamma}\\\mathbb{E}_t[U(C_{s,t+1},\gamma=1)|\tau]=\log(\tau)+\mathbb{E}_t\left[\log[D_te^{\varepsilon_{t+1}}m_t\mathbb{E}[e^{\mu_j}|j\in\mathcal{L}_t]]\right]-\log(1-m_t)\end{gathered}\]

For data users' consumption, assuming the size of data users is $m_t$, each data user $i$ sells $(1-\vartheta)$ capital and retains $\vartheta$ capital, $\vartheta\in(0,1)$. The net income $M_{i,t}$ is:
\[M_{i,t}=\mathbb{E}_t[\pi_{t,t+1}y_{i,t+1}(1-\tau_t)]\]

Where $\pi_{t,t+1}$ is the capital depreciation index, and $\tau_t$ is the data cost rate at stage $t$. Let $N_t^ij$ denote the percentage of firm $i$'s investment in firm $j$'s data capital (risky asset) at time $t$, and $N_it^0$ represent the quantity of risk-free assets. Then the investment budget is:
\[(1-\vartheta)M_{it}=\int_{j\neq i}N_t^{ij}M_{jt}dj+N_{it}^0\]

Specifically, standardizing the stock price and fixing it at 1. Then, for a certain data cost rate $\tau_t=\tau$, the agent's consumption could be written as:
\[C_{i,t+1}=\vartheta y_{i,t+1}(1-\tau_t)+\int_{j\in \mathcal{L}_t}N_t^{ij}y_{j,t}(1-\tau_t)dj+N_{it}^0\]

Under market equilibrium, all firms face identical risk premiums on their assets. Therefore, the optimal investment strategy for entrepreneurs is to allocate funds proportionally to the market value of risk assets. Let the proportion of investment in data capital $N_t^{ij}$ be $\delta(\gamma)$, and the proportion invested in risk-free assets $N_{it}^0$ be $1-\delta(\gamma)$. 

Then, at equilibrium state:
\[\begin{gathered}N_{t}^{ij}M_{jt}=[\delta(\gamma)(1-\vartheta)M_{it}]\times\frac{M_{jt}}{M_{P}}=[\delta(\gamma)(1-\vartheta)M_{it}]\times\frac{M_{jt}}{\int_{p\in \mathcal{L}_{t}}M_{pt}dp}\\N_{t}^{ij}=\frac{\delta(\gamma)(1-\vartheta)M_{it}}{\int_{p\in \mathcal{L}_{t}}M_{pt}dp}\\N_{it}^{0}=(1-\vartheta)M_{it}-\int_{j\neq i}N_{t}^{ij}M_{jt}dj=(1-\vartheta)M_{it}-\frac{[\delta(\gamma)(1-\vartheta)M_{it}]}{\int_{p\in \mathcal{L}_{t}}M_{pt}dp}\int_{j\neq i}M_{jt}dj\end{gathered}\]

According to the \textit{Continuum Hypothesis}:
\[\int_{j\neq i}M_{jt}dj=M_P-M_{jt}\approx M_P\]

Hence:
\[N_{it}^0=(1-\vartheta)M_{it}-\delta(\gamma)(1-\vartheta)M_{it}=[1-\delta(\gamma)](1-\vartheta)M_{it}\]

Under market clearing conditions, where firm $i$'s demand equals firm $j$'s supply: 
\[\begin{gathered}(1-\vartheta)=\int_{j\neq i}N_t^{ij}di=\delta(\gamma)(1-\vartheta)\int_{j\neq i}\frac{M_{it}}{\int_{p\in\mathcal{L}_t}M_{pt}dp}di=\delta(\gamma)(1-\vartheta)\frac{\int_{j\neq i}M_{it}di}{\int_{p\in\mathcal{L}_t}M_{pt}dp}\\\delta(\gamma)=1\end{gathered}\]

That is:
\[\begin{gathered}N_t^{ij}=(1-\vartheta)\frac{e^{\mu_i}}{\int_{p\in\mathcal{L}_t}e^{\mu_p}dp}\\N_{it}^{0}=0\end{gathered}\]

The consumption of data users is:
\[C_{i,t+1}=(1-\tau)D_te^{\mu_i}e^{\varepsilon_{t+1}}[\vartheta e^{\varepsilon_{i,t+1}}+(1-\vartheta)]\]

And the utility function is:
\[\begin{gathered}\mathbb{E}_t[U(C_{i,t+1},\gamma\neq1)|\tau]=\frac{(1-\tau)^{1-\gamma}D_t^{1-\gamma}e^{(1-\gamma)\mu_i}}{1-\gamma}\mathbb{E}_t[e^{(1-\gamma)(\varepsilon_t)}]\mathbb{E}[[\vartheta e^{\varepsilon_{i,t+1}}+(1-\vartheta)]^{1-\gamma}]\\\mathbb{E}_t[U(C_{i,t+1},\gamma=1)|\tau]=\log(1-\tau)+\log[D_te^{\mu_i}]+\mathbb{E}_t\left[\log[e^{\varepsilon_{t+1}}(\vartheta e^{\varepsilon_{i,t+1}}+(1-\vartheta))]\right]\end{gathered}\]

\subsection{Analysis of Heterogeneity}
Define the utility gained by data users through data investment as $V_t^i$, and the utility gained by data providers through payments received for providing data as $V_t^s$. When $V_t^i>V_t^s$, the agent's data investment ability is high, and $V_t^i=V_t^H$. When $V_t^i<V_t^s$, the agent's data investment ability is low, and $V_t^i=V_t^L$. Given a fixed data cost rate $\tau_t=\tau^k$, and assuming equilibrium conditions where the data user set is $\mathcal{L}_k$ with scale $m_t=m^k$, then for agent $i$:
\[V_t^i>V_t^s\]

That is:
\[\begin{gathered}\mathbb{E}_t[U(C_{i,t+1})|\tau^k,m^k]>\mathbb{E}_t[U(C_{s,t+1})|\tau^k,m^k]\\\frac{(1-\tau^k)^{1-\gamma}D_t^{1-\gamma}e^{(1-\gamma)\mu_i}}{1-\gamma}\mathbb{E}_t[e^{(1-\gamma)\varepsilon_{t+1}}]\mathbb{E}[[\vartheta e^{\varepsilon_{i,t+1}}+(1-\vartheta)]^{1-\gamma}]>\begin{bmatrix}\frac{1}{1-\gamma}(\tau^k)^{1-\gamma}D_t^{1-\gamma}\\\times\mathbb{E}[e^{\mu j}|j\in\mathcal{L}_k]^{1-\gamma}\\\times\left(\frac{m^k}{1-m^k}\right)^{1-\gamma}\mathbb{E}_t[e^{(1-\gamma)\varepsilon_t}]\end{bmatrix}\end{gathered}\]

Rearranging and taking the logarithm on both sides of the inequality :
\[\begin{gathered}\mu_i+\frac{1}{1-\gamma}\log(\mathbb{E}[[\vartheta e^{\varepsilon_{i,t+1}}+(1-\vartheta)]^{1-\gamma}])>\begin{bmatrix}\log\left(\frac{\tau^k}{1-\tau^k}\right)+\log(\mathbb{E}[e^{\mu_j}|j\in\mathcal{L}_k])\\+\log\left(\frac{m^k}{1-m^k}\right)\end{bmatrix}\\\mu_i>K(k)=\begin{bmatrix}\log\left(\frac{\tau^k}{1-\tau_k}\right)+\log(\mathbb{E}[e^{\mu_j}|j\in\mathcal{L}_k])\\+\log\left(\frac{m^k}{1-m^k}\right)-\left(\frac{1}{1-\gamma}\log(\mathbb{E}[[\vartheta e^{\varepsilon_{i,t+1}}+(1-\vartheta)]^{1-\gamma}])\right)\end{bmatrix}\end{gathered}\]

Based on the definition of $m^k$ and the distribution of the agent's data investment ability $\mu_i{\sim}\mathcal{N}(\bar{\mu},\sigma_\mu^2)$:
\[\begin{gathered}\begin{gathered}m_{t}^{k}=\int_{K(k)}^{\infty}\phi\left(\mu_{i};\overline{\mu},\sigma_{\mu}^{2}\right)d\mu_{i}=1-\Phi\left(K(k);\overline{\mu},\sigma_{\mu}^{2}\right)\\\mathbb{E}[e^{\mu_j}|j\in\mathcal{L}_k]=\frac{1}{m_{t}^{k}}\int_{K(k)}^{\infty}e^{\mu_{j}}\phi(\mu_{j};\overline{\mu},\sigma_{\mu}^{2})d\mu_{j}=\frac{e^{\overline{\mu}+\frac{1}{2}\sigma_{\mu}^{2}}\left(1-\Phi(K(k);\overline{\mu}+\sigma_{\mu}^{2},\sigma_{\mu}^{2})\right)}{m_{t}^{k}}\\K(k)=\begin{bmatrix}\log\left(\frac{\tau^k}{1-\tau^k}\right)+\overline{\mu}+\frac{1}{2}\sigma_\mu^2+\log\left(\frac{1-\Phi\left(K(k);\overline{\mu}+\sigma_\mu^2,\sigma_\mu^2\right)}{\Phi\left(\mu_k;\overline{\mu},\sigma_\mu^2\right)}\right)\\-\frac{1}{1-\gamma}\log(\mathbb{E}[[\vartheta e^{\varepsilon_{i,t+1}}+(1-\vartheta)]^{1-\gamma}])\end{bmatrix}\end{gathered}\end{gathered}\]

Then, we define $\mu_k=K(k)-\overline{\mu}$: 
\[\mu_{k}=\log\left(\frac{\tau^{k}}{1-\tau^{k}}\right)+\frac{1}{2}\sigma_{\mu}^{2}+\log\left(\frac{1-\Phi(\mu_{k};\sigma_{\mu}^{2},\sigma_{\mu}^{2})}{\Phi(\mu_{k};0,\sigma_{\mu}^{2})}\right)-\frac{1}{1-\gamma}\log(\mathbb{E}[[\vartheta e^{\varepsilon_{i,t+1}}+(1-\vartheta)]^{1-\gamma}])\]

Therefore:
\[m_t^k=\int_{K(k)}^\infty\phi\left(\mu_i;\overline{\mu},\sigma_\mu^2\right)d\mu_i=1-\Phi\left(K(k)-\overline{\mu};0,\sigma_\mu^2\right)=1-\Phi\left(\mu_k;0,\sigma_\mu^2\right)\]

Similarly, for the case where $\gamma=1$, we can obtain a similar result as follows:
\[\mu_k=\log\left(\frac{\tau^k}{1-\tau^k}\right)+\frac{1}{2}\sigma_\mu^2+\log\left(\frac{1-\Phi(\mu_k;\sigma_\mu^2,\sigma_\mu^2)}{\Phi(\mu_k;0,\sigma_\mu^2)}\right)-\mathbb{E}_t[\log(\vartheta e^{\varepsilon_{i,t+1}}+(1-\vartheta))]\]

Now, we set:
\[\mu_k=F(\tau^k,\mu_k)\]

Then:
\[d\mu_k=\frac{\partial F(\tau^k,\mu_k)}{\partial\tau^k}d\tau^k+\frac{\partial F(\tau^k,\mu_k)}{\partial\mu_k}d\mu_k\]

Hence:
\[\frac{d\mu_k}{d\tau^k}=\frac{\frac{\partial F(\tau^k,K(k))}{\partial\tau^k}}{1-\frac{\partial F(\tau^k,K(k))}{\partial K(k)}}>0\]

Therefore, since $(d\mu_k)/(d\tau^k)>0$, the threshold for data investment ability and data costs exhibit a positive correlation, indicating that agents with higher data investment ability incur greater data costs. Moreover, based on $m_t^k =1-\Phi(\mu_k;0,\sigma_{\mu}^2)$, the data investment ability threshold $F(\tau^k,\mu_k)$ also determines whether an agent belongs to the data user group $m_t^k$, where $m_t^k\subset m_t$.

This paper identifies two types of agents with two kinds of data investment ability: When $\mu_H=F(\tau^H,\mu_H)>\mu_k$, the agent possesses high data investment ability. Ideally, they should act as data users to achieve satisfactory utility, and will choose a higher data cost rate $\tau^H$. When $\mu_L=F(\tau^L,\mu_L)<\mu_k$, the agent possesses low data investment ability, and ideally should act as a data provider to achieve satisfactory utility. This paper assumes that the choice of agents with low data investment ability is imperfect. That is, firms with low data investment ability will not immediately switch to being data providers due to the low returns and utility gained from data investment. Instead, they choose to continue as data users and bear the economic consequences resulting from their low data investment ability ($V_t^L<V_t^s$). Such agents will select a lower data cost rate $\tau^L$. Accordingly, this paper posits that differences in firms' data investment ability lead to heterogeneity in data input scale $d_i$, output $y_i$, technology $z_i$ (total factor productivity), and financial friction $\lambda_i$, where: $\mu_i\in\{\mu_H,\mu_L\},\mathbb{F}:\mu_i\to\{y_i,z_i,d_i,\lambda_i\}$. We accordingly propose Theorem 1:

\begin{theorem}
\textit{Heterogeneous firms are classified into High-type and Low-type groups based on the data investment ability thresholds $\mu_k$. High-type firms possess data investment ability $\mu_H(\mu_H>\mu_k)$, and $\mathbb{F}:\mu_H\to\{y_H,z_H,d_H,\lambda_H\}$. Low-type firms possess data investment ability $\mu_L(\mu_L<\mu_k)$, where $\mathbb{F}:\mu_L\to\{y_L,z_L,d_L,\lambda_L\}$. Here, $y_H>y_L,z_H>z_L,d_H>d_L,\lambda_H>\lambda_L$.}
\end{theorem}

\begin{proof}
Based on the theory above, we assume a linear relationship between data input costs and data input scale. That is, the enterprises with higher data investment ability have reason to bear higher costs $\tau^H$ to purchase and input more data elements $d_H$, while the enterprises with lower data investment ability choose to bear the economic consequences of data usage, and the consequences will prevent these firms from achieving satisfactory returns and utility. Consequently, they prefer to incur lower costs $\tau^L$ to acquire and utilize data elements $d_L$. Since $\tau^H>\tau^k>\tau^L$, it follows that $d_H(\tau^H)>d_k(\tau^k)>d_L(\tau^L)$. Then, according to \citet{jones2020nonrivalry}, assuming an exponential relationship $\eta$ between data $d_i$ and technology $z_i$, where $z_i=d_i^{\eta},\eta\in(0,1)$, it follows that $z_H(d_H)>z_L(d_L)$.

Due to:
\[y_t=\int_{j\in \mathcal{L}_t}y_{j,t}=D_te^{\varepsilon_{t+1}}m_t\mathbb{E}[e^{\mu_j}|j\in \mathcal{L}_t]=D_te^{\varepsilon_{t+1}}m_te^{\overline{\mu}+\frac{1}{2}\sigma_\mu^2}\frac{\left(1-\Phi(K(k);\overline{\mu}+\sigma_\mu^2,\sigma_\mu^2)\right)}{1-\Phi(K(k);\overline{\mu},\sigma_\mu^2)}\]

And:
\[\mu_k=K(k)-\overline{\mu}\]

Then, we have:
\[y_t(\mu_k)=D_te^{\varepsilon_{t+1}}m_te^{\overline{\mu}+\frac{1}{2}\sigma_\mu^2}\frac{\left(1-\Phi(\mu_k;\sigma_\mu^2,\sigma_\mu^2)\right)}{1-\Phi(\mu_k;0,\sigma_\mu^2)}\]

Denote $f(\mu_k)=\frac{\left(1-\Phi(\mu_k;\sigma_\mu^2,\sigma_\mu^2)\right)}{1-\Phi(\mu_k;0,\sigma_\mu^2)}$, we get:
\[\frac{\partial f(\mu_k)}{\partial\mu_k}=\frac{-\phi(\mu_k;\sigma_\mu^2,\sigma_\mu^2)[1-\Phi(\mu_k;0,\sigma_\mu^2)]+[1-\Phi(\mu_k;\sigma_\mu^2,\sigma_\mu^2)]\phi(\mu_k;0,\sigma_\mu^2)}{\left[1-\Phi(\mu_k;0,\sigma_\mu^2)\right]^2}\]

That is:
\[\frac{\phi(\mu_k;0,\sigma_\mu^2)}{1-\Phi(\mu_k;0,\sigma_\mu^2)}>\frac{\phi(\mu_k-\sigma_\mu^2;0,\sigma_\mu^2)}{1-\Phi(\mu_k-\sigma_\mu^2;0,\sigma_\mu^2)}=\frac{\phi(\mu_k;\sigma_\mu^2,\sigma_\mu^2)}{1-\Phi(\mu_k;\sigma_\mu^2,\sigma_\mu^2)}\]

The inequality above holds for all cases. Therefore, for $\mu_H>\mu_k>\mu_L$, the total output $y_t(\mu_H)$ of enterprises with higher data investment ability will exceed total output $y_t(\mu_L)$ of enterprises with lower data investment ability, i.e., $y_H>y_L$.

According to \citet{moll2014productivity}, the amount of capital available to firms for production is jointly constrained by their own assets and the level of financial market development. Let $\lambda$ denote the borrowing constraint coefficient, representing financial frictions. For $\lambda\geq1$, as $\lambda\to+\infty$, the financial market approaches a perfect state where borrowing resistance for firms becomes negligible. When $\lambda=1$, the financial market is completely closed, and firms must provide all capital required for production themselves. When $\lambda\in(1,+\infty)$, firms face an upper limit on borrowing capacity, specifically $\lambda$ times their individual net wealth. $W_t$ denotes the entrepreneur's net wealth. This paper assumes the initial net wealth $W_0$ originates solely from assets held at the start of the enterprise investment, i.e., $W_0=y_{i,0}(1-\tau_0)$. $K_t$ represents the scale of capital the entrepreneur can borrow.

The enterprise allocates its wealth $W_t$ across two types of assets: The risk-free asset accounts for a proportion of $(1-\alpha),\alpha\in(0,1)$ with a return rate of $r_fdt$, where $r_f=\log R_f$ and $R_f$ represents the risk-free asset return. The risk asset proportion is $\alpha$, with a return of $(r_f+\hat{\mu})dt$. The risk asset also includes a \textit{Brownian Motion} term $\sigma dZ$. When a \textit{Poisson} event occurs, the asset loss ratio is $L$, resulting in a \textit{Poisson Jump} term of $(-LdN)$, where $\mathbb{E}[dN]=wdt$. Since the overall expected excess return $\hat{\mu}$ should equal the sum of the continuous portion of the return and the expected compensation for \textit{Poisson Jump}, we have $\hat{\mu}=\mu_0+w\mathbb{E}L$, where $\mu_0=\log(\mathbb{E}R/R_f)>0$ represents the continuous portion of the excess return. The enterprise invests capital factors at a rate $K_t$, constrained by financial friction $\lambda$. Hence, $K_t=\lambda W_t$. Therefore, the enterprise's wealth dynamics could be written as:
\[dW_t=(1-\alpha)W_tr_fdt+\alpha W_t[(r_f+\hat{\mu})dt+\sigma dZ-LdN]-\lambda W_tdt\]

Hence:
\[\begin{gathered}dW_t=W_t[r_f+\alpha\hat{\mu}-\lambda]dt+\alpha W_t\sigma dZ-\alpha W_tLdN\\\frac{dK}{K}=\frac{dW}{W}=\underbrace{[r_f+\alpha\hat{\mu}-\lambda]dt+\alpha\sigma dZ}_{Diffusion}\underbrace{-\alpha LdN}_{Poisson~Jump}\end{gathered}\]

Denote $f(K_t)=\log K_t$. According to \textit{Ito's} lemma:
\[df(K_t)=\underbrace{\frac{\partial f}{\partial K}dK+\frac{1}{2}\frac{\partial^2f}{\partial K^2}(dK)^2}_{Diffusion}+\underbrace{\Delta f(K_t)}_{Poisson~Jump}\]

And:
\[\begin{gathered}dK_{Diffusion}=K\left((r_{f}+\hat{\mu}-\lambda)dt+\alpha\sigma dZ\right)\\\frac{\partial f}{\partial K}dK_{Diffusion}=(r_f+\hat{\mu}-\lambda)dt+\alpha\sigma d\\\frac{1}{2}\frac{\partial^2f}{\partial K^2}\left(dK_{Diffusion}\right)^2=-\frac{1}{2}\alpha^2\sigma^2dt\end{gathered}\]

Due to:
\[K_t=K_{t-1}(1-\alpha L)\]

Then:
\[\begin{gathered}\Delta f(K_t)=\mathrm{~log}K_t-\mathrm{log}K_{t-1}=\log\left(\frac{K_t}{K_{t-1}}\right)=\log(1-\alpha L)\\d\mathrm{log}K=\left(r_f+\alpha\hat{\mu}-\frac{1}{2}\alpha^2\sigma^2-\lambda\right)dt+\alpha\sigma dZ+\log(1-\alpha L)dN\end{gathered}\]

Solve the \textit{SDE}, we get:
\[\begin{gathered}\mathrm{log}K_t-\mathrm{log}K_0=\int_0^t\left(r_f+\alpha\hat{\mu}-\frac{1}{2}\alpha^2\sigma^2-\lambda\right)ds+\int_0^t\alpha\sigma dZ_s+\int_0^t\log(1-\alpha L)dN_s\\\mathrm{log}K_t=\mathrm{~log}K_0+\left(r_f+\alpha\hat{\mu}-\frac{1}{2}\alpha^2\sigma^2-\lambda\right)t+\alpha\sigma dZ_t+\sum_{i=1}^{N_t}\log(1-\alpha L_i)\\K_t=\lambda W_0\mathrm{exp}\left\{\left(r_f+\alpha\hat{\mu}-\frac{1}{2}\alpha^2\sigma^2-\lambda\right)t+\alpha\sigma dZ_t\right\}\prod_{i=1}^{N_t}(1-\alpha L_i)\end{gathered}\]

Hence:
\[\mathbb{E}K_t=\lambda W_0\mathrm{exp}\left\{\left(r_f+\alpha\hat{\mu}-\frac{1}{2}\alpha^2\sigma^2-\lambda\right)t\right\}\times\mathbb{E}\mathrm{exp}(\alpha\sigma dZ_t)\times\mathbb{E}\prod_{i=1}^{N_t}(1-\alpha L_i)\]

According to the mathematics condition:
\[\begin{gathered}\mathbb{E}\mathrm{exp}(\alpha\sigma dZ_t)=\exp\left(\frac{1}{2}\alpha^2\sigma^2t\right),Z_t{\sim}N(0,t)\\\mathbb{E}\left[\prod_{i=1}^{N_t}(1-\alpha L_i)|N_t=n\right]=[\mathbb{E}(1-\alpha L)]^n\\\mathbb{E}\prod_{i=1}^{N_t}(1-\alpha L_i)=\sum_{n=0}^\infty\mathbb{P}(N_t=n)\times[\mathbb{E}(1-\alpha L)]^n,N_t\sim\mathrm{Poisson}(wt)\\\mathbb{E}\prod_{i=1}^{N_t}(1-\alpha L_i)=\sum_{n=0}^\infty e^{-wt}\frac{(wt)^n}{n!}\times[\mathbb{E}(1-\alpha L)]^n=\exp\{wt[\mathbb{E}(1-\alpha L)-1]\}\end{gathered}\]

That is, $\mathbb{E}K_t$ could be rewritten as:
\[\mathbb{E}K_t=\lambda W_0\mathrm{exp}\{(r_f+\alpha\hat{\mu}-\lambda+w[\mathbb{E}(1-\alpha L)-1])t\}\]

Then, we get:
\[\begin{gathered}\frac{e^{\lambda t}}{\lambda}=W_0(\mathbb{E}K_t)^{-1}e^{\{(r_f+\alpha\widehat{\mu}+w[\mathbb{E}(1-\alpha L)-1])t\}}\\\frac{e^{\lambda t}}{\lambda}=e^{\mu_i+\varepsilon_0+\varepsilon_{i,0}}D_0(1-\tau_0)(\mathbb{E}K_t)^{-1}e^{\{\left(r_f+\alpha\widehat{\mu}+w[\mathbb{E}(1-\alpha L)-1]\right)t\}}\end{gathered}\]

We denote:
\[\begin{gathered}f(\lambda,t)=\frac{e^{\lambda t}}{\lambda}\\f(\mu_i,t)=e^{\mu_i+\varepsilon_0+\varepsilon_{i,0}}D_0(1-\tau_0)(\mathbb{E}K_t)^{-1}e^{\{\left(r_f+\alpha\widehat{\mu}+w[\mathbb{E}(1-\alpha L)-1]\right)t\}}\end{gathered}\]

When $t=t^*$
\[\begin{gathered}\frac{\partial f(\lambda,t^*)}{\partial\lambda}>0,\lambda\geq1,t^*>0\\\frac{\partial f(\mu_i,t^*)}{\partial\mu_i}>0,t^*>0\end{gathered}\]

\begin{figure}[htbp]
\centering
\includegraphics[width=10cm]{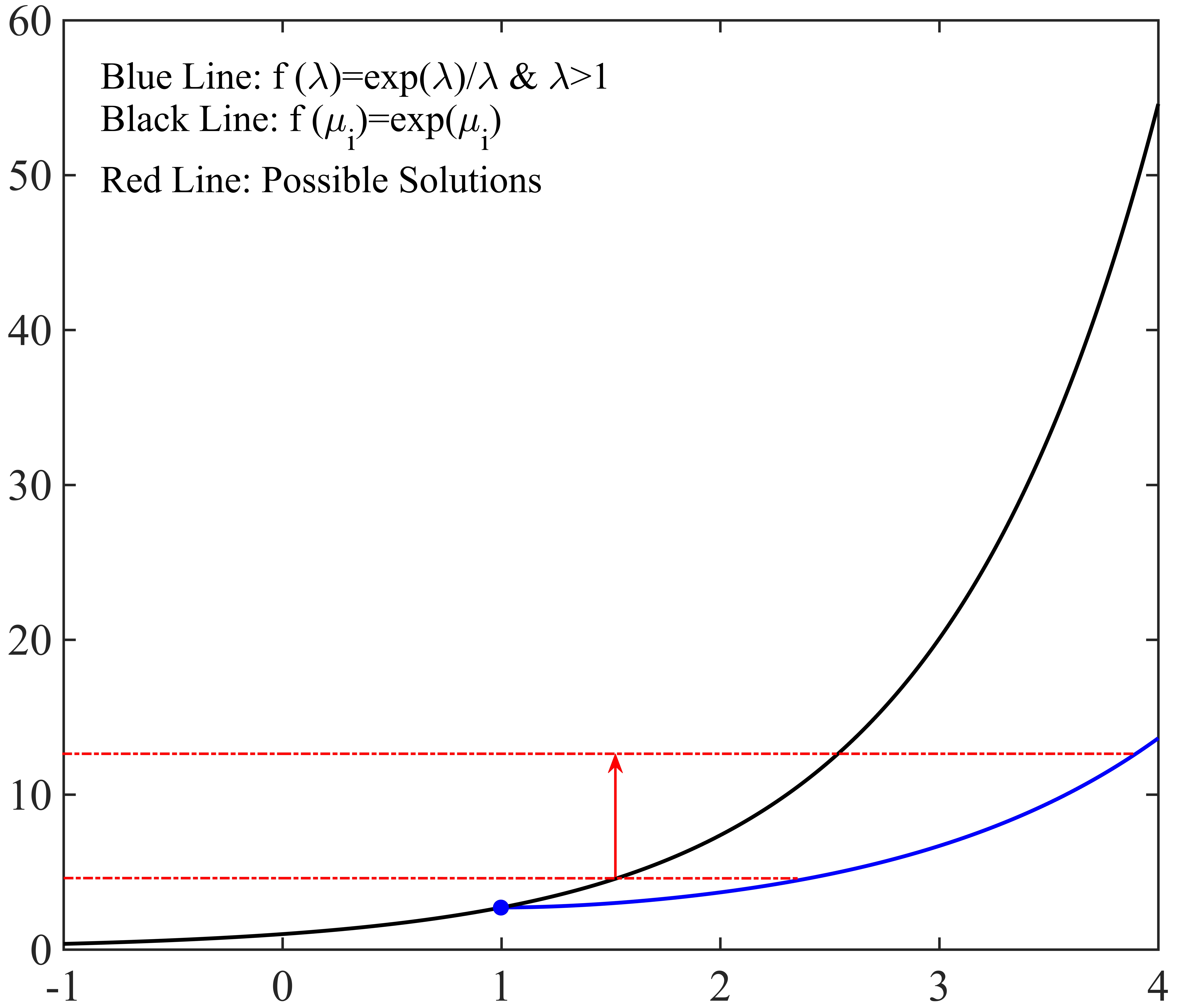}
\caption{\label{fig:F1}Possible Solution for $f(\lambda,t^*)=f(\mu_i,t^*)$}
\end{figure}

We constructed a special case to derive the analytical solution. As shown in Figure 1, when time $t^*$ is fixed and the equation's solutions exist, increasing the agent's ability $\mu_i$ leads to higher total output, which also leads the increase of $\lambda$ and reduces financial friction. Therefore, for firms with high data investment ability $\mu_H$, the financial friction coefficient is $\lambda_H$, and for firms with low data investment ability $\mu_L$, it is $\lambda_L$. Since $\mu_H>\mu_k>\mu_L$, it follows that $\lambda_H>\lambda_L>1$.
\end{proof}

\section{Conclusion}
In this brief analytical paper, our objective is to examine potential heterogeneity in the data economy by analyzing agents' varying ability to utilize data. We categorize agents in the data economy as either data users or data providers. Data users form data capital through utilizing providers' data for investment purposes and pay providers for data costs, which we define as data investment ability. Theoretical analysis indicates that : Agents with higher data investment ability can acquire greater utility by investing in data. Therefore, they have reason to pay higher data costs to obtain a larger scale of data. Furthermore, this leads agents with higher data investment ability to achieve greater productivity, faster technological progress, and face lower financial frictions.

\singlespacing
\setlength{\bibsep}{2pt}
\bibliographystyle{plainnat}
\bibliography{main}

\begin{thebibliography}{67}
\providecommand{\natexlab}[1]{#1}
\providecommand{\url}[1]{\texttt{#1}}
\expandafter\ifx\csname urlstyle\endcsname\relax
  \providecommand{\doi}[1]{doi: #1}\else
  \providecommand{\doi}{doi: \begingroup \urlstyle{rm}\Url}\fi

\bibitem[Abis and Veldkamp(2024)]{abis2024changing}
Simona Abis and Laura Veldkamp.
\newblock The changing economics of knowledge production.
\newblock \emph{The Review of Financial Studies}, 37\penalty0 (1):\penalty0 89--118, 2024.

\bibitem[Acemoglu(2025)]{acemoglu2025simple}
Daron Acemoglu.
\newblock The simple macroeconomics of ai.
\newblock \emph{Economic Policy}, 40\penalty0 (121):\penalty0 13--58, 2025.

\bibitem[Acemoglu and Restrepo(2018)]{acemoglu2018race}
Daron Acemoglu and Pascual Restrepo.
\newblock The race between man and machine: Implications of technology for growth, factor shares, and employment.
\newblock \emph{American Economic Review}, 108\penalty0 (6):\penalty0 1488--1542, 2018.

\bibitem[Acemoglu et~al.(2024)Acemoglu, Ozdaglar, and Siderius]{acemoglu2024model}
Daron Acemoglu, Asuman Ozdaglar, and James Siderius.
\newblock A model of online misinformation.
\newblock \emph{Review of Economic Studies}, 91\penalty0 (6):\penalty0 3117--3150, 2024.

\bibitem[Acemoglu et~al.(2025)Acemoglu, Makhdoumi, Malekian, and Ozdaglar]{acemoglu2025big}
Daron Acemoglu, Ali Makhdoumi, Azarakhsh Malekian, and Asuman Ozdaglar.
\newblock When big data enables behavioral manipulation.
\newblock \emph{American Economic Review: Insights}, 7\penalty0 (1):\penalty0 19--38, 2025.

\bibitem[Achdou and Lauri{\`e}re(2020)]{achdou2020mean}
Yves Achdou and Mathieu Lauri{\`e}re.
\newblock Mean field games and applications: numerical aspects.
\newblock \emph{arXiv preprint arXiv:2003.04444}, 2020.

\bibitem[Achdou et~al.(2022)Achdou, Han, Lasry, Lions, and Moll]{achdou2022income}
Yves Achdou, Jiequn Han, Jean-Michel Lasry, Pierre-Louis Lions, and Benjamin Moll.
\newblock Income and wealth distribution in macroeconomics: A continuous-time approach.
\newblock \emph{The Review of Economic Studies}, 89\penalty0 (1):\penalty0 45--86, 2022.

\bibitem[Adrian and Shin(2010)]{adrian2010liquidity}
Tobias Adrian and Hyun~Song Shin.
\newblock Liquidity and leverage.
\newblock \emph{Journal of Financial Intermediation}, 19\penalty0 (3):\penalty0 418--437, 2010.

\bibitem[Ahn et~al.(2018)Ahn, Kaplan, Moll, Winberry, and Wolf]{ahn2018inequality}
SeHyoun Ahn, Greg Kaplan, Benjamin Moll, Thomas Winberry, and Christian Wolf.
\newblock When inequality matters for macro and macro matters for inequality.
\newblock \emph{NBER Macroeconomics Annual}, 32\penalty0 (1):\penalty0 1--75, 2018.

\bibitem[Aiyagari(1994)]{aiyagari1994uninsured}
S~Rao Aiyagari.
\newblock Uninsured idiosyncratic risk and aggregate saving.
\newblock \emph{The Quarterly Journal of Economics}, 109\penalty0 (3):\penalty0 659--684, 1994.

\bibitem[Angeletos and Sastry(2025)]{angeletos2025inattentive}
George-Marios Angeletos and Karthik~A Sastry.
\newblock Inattentive economies.
\newblock \emph{Journal of Political Economy}, 133\penalty0 (7):\penalty0 2265--2319, 2025.

\bibitem[Angeletos et~al.(2018)Angeletos, Collard, and Dellas]{angeletos2018quantifying}
George-Marios Angeletos, Fabrice Collard, and Harris Dellas.
\newblock Quantifying confidence.
\newblock \emph{Econometrica}, 86\penalty0 (5):\penalty0 1689--1726, 2018.

\bibitem[Athey(2025)]{athey2025presidential}
Susan Athey.
\newblock Presidential address: The economist as designer in the innovation process for socially impactful digital products.
\newblock \emph{American Economic Review}, 115\penalty0 (4):\penalty0 1059--1099, 2025.

\bibitem[Aum and Shin(2025)]{aum2025labor}
Sangmin Aum and Yongseok Shin.
\newblock The labor market impact of digital technologies.
\newblock Working Paper 33469, NBER, 2025.

\bibitem[Beraja and Zorzi(2025)]{beraja2025inefficient}
Martin Beraja and Nathan Zorzi.
\newblock Inefficient automation.
\newblock \emph{Review of Economic Studies}, 92\penalty0 (1):\penalty0 69--96, 2025.

\bibitem[Bessen(2019)]{bessen2019automation}
James Bessen.
\newblock Automation and jobs: When technology boosts employment.
\newblock \emph{Economic Policy}, 34\penalty0 (100):\penalty0 589--626, 2019.

\bibitem[Bhandari et~al.(2025)Bhandari, Borovi{\v{c}}ka, and Ho]{bhandari2025survey}
Anmol Bhandari, Jaroslav Borovi{\v{c}}ka, and Paul Ho.
\newblock Survey data and subjective beliefs in business cycle models.
\newblock \emph{Review of Economic Studies}, 92\penalty0 (3):\penalty0 1375--1437, 2025.

\bibitem[Bilal(2023)]{bilal2023solving}
Adrien Bilal.
\newblock Solving heterogeneous agent models with the master equation.
\newblock Working Paper 31103, NBER, 2023.

\bibitem[Brunnermeier and Sannikov(2014)]{brunnermeier2014macroeconomic}
Markus~K Brunnermeier and Yuliy Sannikov.
\newblock A macroeconomic model with a financial sector.
\newblock \emph{American Economic Review}, 104\penalty0 (2):\penalty0 379--421, 2014.

\bibitem[Brynjolfsson et~al.(2025)Brynjolfsson, Li, and Raymond]{brynjolfsson2025generative}
Erik Brynjolfsson, Danielle Li, and Lindsey Raymond.
\newblock Generative ai at work.
\newblock \emph{The Quarterly Journal of Economics}, 140\penalty0 (2):\penalty0 889--942, 2025.

\bibitem[Caplin(2025)]{caplin2025data}
Andrew Caplin.
\newblock Data engineering for cognitive economics.
\newblock \emph{Journal of Economic Literature}, 63\penalty0 (1):\penalty0 164--196, 2025.

\bibitem[Caplin et~al.(2022)Caplin, Dean, and Leahy]{caplin2022rationally}
Andrew Caplin, Mark Dean, and John Leahy.
\newblock Rationally inattentive behavior: Characterizing and generalizing shannon entropy.
\newblock \emph{Journal of Political Economy}, 130\penalty0 (6):\penalty0 1676--1715, 2022.

\bibitem[Cong and Mayer(2023)]{cong2023data}
Lin~William Cong and Simon Mayer.
\newblock Data union and regulation in a data economy.
\newblock Working Paper 30881, NBER, 2023.

\bibitem[Cong et~al.(2021)Cong, Xie, and Zhang]{cong2021knowledge}
Lin~William Cong, Danxia Xie, and Longtian Zhang.
\newblock Knowledge accumulation, privacy, and growth in a data economy.
\newblock \emph{Management Science}, 67\penalty0 (10):\penalty0 6480--6492, 2021.

\bibitem[Corrado et~al.(2024)Corrado, Haskel, Iommi, Jona-Lasinio, and Bontadini]{corrado2024data}
Carol Corrado, Jonathan Haskel, Massimiliano Iommi, Cecilia Jona-Lasinio, and Filippo Bontadini.
\newblock Data, intangible capital, and productivity.
\newblock \emph{NBER Chapters}, 2024.

\bibitem[Costa and Aparicio(2025)]{costa2025exploring}
Carlos~J Costa and Joao~Tiago Aparicio.
\newblock Exploring the societal and economic impacts of artificial intelligence: A scenario generation methodology.
\newblock \emph{arXiv preprint arXiv:2504.01992}, 2025.

\bibitem[Costa et~al.(2024)Costa, Aparicio, and Aparicio]{costa2024socio}
Carlos~J Costa, Joao~Tiago Aparicio, and Manuela Aparicio.
\newblock Socio-economic consequences of generative ai: A review of methodological approaches.
\newblock \emph{arXiv preprint arXiv:2411.09313}, 2024.

\bibitem[Ding et~al.(2025)Ding, Jiang, Msall, and Notowidigdo]{ding2025consumer}
Jing Ding, Lei Jiang, Lucy Msall, and Matthew~J Notowidigdo.
\newblock Consumer-financed fiscal stimulus: Evidence from digital coupons in china.
\newblock \emph{American Economic Review: Insights}, 7\penalty0 (3):\penalty0 411--427, 2025.

\bibitem[Einav and Levin(2014)]{einav2014economics}
Liran Einav and Jonathan Levin.
\newblock Economics in the age of big data.
\newblock \emph{Science}, 346\penalty0 (6210), 2014.

\bibitem[Farboodi and Veldkamp(2021)]{farboodi2021model}
Maryam Farboodi and Laura Veldkamp.
\newblock A model of the data economy.
\newblock Working Paper 28427, NBER, 2021.

\bibitem[Farboodi et~al.(2019)Farboodi, Mihet, Philippon, and Veldkamp]{farboodi2019big}
Maryam Farboodi, Roxana Mihet, Thomas Philippon, and Laura Veldkamp.
\newblock Big data and firm dynamics.
\newblock In \emph{AEA papers and proceedings}, volume 109, pages 38--42, 2019.

\bibitem[Fern{\'a}ndez-Villaverde et~al.(2023)Fern{\'a}ndez-Villaverde, Hurtado, and Nuno]{fernandez2023financial}
Jes{\'u}s Fern{\'a}ndez-Villaverde, Samuel Hurtado, and Galo Nuno.
\newblock Financial frictions and the wealth distribution.
\newblock \emph{Econometrica}, 91\penalty0 (3):\penalty0 869--901, 2023.

\bibitem[Gabaix et~al.(2016)Gabaix, Lasry, Lions, and Moll]{gabaix2016dynamics}
Xavier Gabaix, Jean-Michel Lasry, Pierre-Louis Lions, and Benjamin Moll.
\newblock The dynamics of inequality.
\newblock \emph{Econometrica}, 84\penalty0 (6):\penalty0 2071--2111, 2016.

\bibitem[Gans(2025{\natexlab{a}})]{gans2025a}
Joshua~S Gans.
\newblock A quest for ai knowledge.
\newblock Working Paper 33566, NBER, 2025{\natexlab{a}}.

\bibitem[Gans(2025{\natexlab{b}})]{gans2025b}
Joshua~S Gans.
\newblock Ai as strategist.
\newblock Working Paper 33650, NBER, 2025{\natexlab{b}}.

\bibitem[Glaeser et~al.(2018)Glaeser, Kominers, Luca, and Naik]{glaeser2018big}
Edward~L Glaeser, Scott~Duke Kominers, Michael Luca, and Nikhil Naik.
\newblock Big data and big cities: The promises and limitations of improved measures of urban life.
\newblock \emph{Economic Inquiry}, 56\penalty0 (1):\penalty0 114--137, 2018.

\bibitem[Goldfarb and Tucker(2019)]{goldfarb2019digital}
Avi Goldfarb and Catherine Tucker.
\newblock Digital economics.
\newblock \emph{Journal of Economic Literature}, 57\penalty0 (1):\penalty0 3--43, 2019.

\bibitem[Graetz and Michaels(2018)]{graetz2018robots}
Georg Graetz and Guy Michaels.
\newblock Robots at work.
\newblock \emph{Review of Economics and Statistics}, 100\penalty0 (5):\penalty0 753--768, 2018.

\bibitem[Hampole et~al.(2025)Hampole, Papanikolaou, Schmidt, and Seegmiller]{hampole2025artificial}
Menaka Hampole, Dimitris Papanikolaou, Lawrence~DW Schmidt, and Bryan Seegmiller.
\newblock Artificial intelligence and the labor market.
\newblock Working Paper 33509, NBER, 2025.

\bibitem[Hayek(1945)]{hayek1945use}
FA~Hayek.
\newblock The use of knowledge in society fa hayek.
\newblock \emph{American Economic Review}, 35\penalty0 (4):\penalty0 519--530, 1945.

\bibitem[H{\'e}mous and Olsen(2022)]{hemous2022rise}
David H{\'e}mous and Morten Olsen.
\newblock The rise of the machines: Automation, horizontal innovation, and income inequality.
\newblock \emph{American Economic Journal: Macroeconomics}, 14\penalty0 (1):\penalty0 179--223, 2022.

\bibitem[Hu(2025{\natexlab{a}})]{hu2025analysis}
Yongheng Hu.
\newblock Analysis theory of data economy: Dataization, technological progress and dynamic general equilibrium.
\newblock \emph{arXiv preprint arXiv:2507.13274}, 2025{\natexlab{a}}.

\bibitem[Hu(2025{\natexlab{b}})]{hu2025big}
Yongheng Hu.
\newblock How big data dilutes cognitive resources, interferes with rational decision-making and affects wealth distribution?
\newblock \emph{arXiv preprint arXiv:2508.20435}, 2025{\natexlab{b}}.

\bibitem[Huggett(1993)]{huggett1993risk}
Mark Huggett.
\newblock The risk-free rate in heterogeneous-agent incomplete-insurance economies.
\newblock \emph{Journal of economic Dynamics and Control}, 17\penalty0 (5-6):\penalty0 953--969, 1993.

\bibitem[Jiang et~al.(2025)Jiang, Park, Xiao, and Zhang]{jiang2025ai}
Wei Jiang, Junyoung Park, Rachel~Jiqiu Xiao, and Shen Zhang.
\newblock Ai and the extended workday: Productivity, contracting efficiency, and distribution of rents.
\newblock Working Paper 33536, NBER, 2025.

\bibitem[Jones(2024)]{jones2024ai}
Charles~I Jones.
\newblock The ai dilemma: Growth versus existential risk.
\newblock \emph{American Economic Review: Insights}, 6\penalty0 (4):\penalty0 575--590, 2024.

\bibitem[Jones(2025)]{jones2025much}
Charles~I Jones.
\newblock How much should we spend to reduce ai's existential risk?
\newblock Working Paper 33602, NBER, 2025.

\bibitem[Jones and Tonetti(2020)]{jones2020nonrivalry}
Charles~I Jones and Christopher Tonetti.
\newblock Nonrivalry and the economics of data.
\newblock \emph{American Economic Review}, 110\penalty0 (9):\penalty0 2819--2858, 2020.

\bibitem[Kaplan et~al.(2018)Kaplan, Moll, and Violante]{kaplan2018monetary}
Greg Kaplan, Benjamin Moll, and Giovanni~L Violante.
\newblock Monetary policy according to hank.
\newblock \emph{American Economic Review}, 108\penalty0 (3):\penalty0 697--743, 2018.

\bibitem[Krusell and Smith(1998)]{krusell1998income}
Per Krusell and Anthony~A Smith, Jr.
\newblock Income and wealth heterogeneity in the macroeconomy.
\newblock \emph{Journal of Political Economy}, 106\penalty0 (5):\penalty0 867--896, 1998.

\bibitem[Lazer et~al.(2014)Lazer, Kennedy, King, and Vespignani]{lazer2014parable}
David Lazer, Ryan Kennedy, Gary King, and Alessandro Vespignani.
\newblock The parable of google flu: traps in big data analysis.
\newblock \emph{Science}, 343\penalty0 (6176):\penalty0 1203--1205, 2014.

\bibitem[Le~Grand and Ragot(2025)]{le2025optimal}
Fran{\c{c}}ois Le~Grand and Xavier Ragot.
\newblock Optimal fiscal policy with heterogeneous agents and capital: Should we increase or decrease public debt and capital taxes?
\newblock \emph{Journal of Political Economy}, 133\penalty0 (7), 2025.

\bibitem[Lucas~Jr(1988)]{lucas1988mechanics}
Robert~E Lucas~Jr.
\newblock On the mechanics of economic development.
\newblock \emph{Journal of Monetary Economics}, 22\penalty0 (1):\penalty0 3--42, 1988.

\bibitem[Lynch(2008)]{lynch2008your}
Clifford Lynch.
\newblock How do your data grow?
\newblock \emph{Nature}, 455\penalty0 (7209):\penalty0 28--29, 2008.

\bibitem[Madsen and Vellodi(2025)]{madsen2025insider}
Erik Madsen and Nikhil Vellodi.
\newblock Insider imitation.
\newblock \emph{Journal of Political Economy}, 133\penalty0 (2):\penalty0 652--709, 2025.

\bibitem[Maxted et~al.(2025)Maxted, Laibson, and Moll]{maxted2025present}
Peter Maxted, David Laibson, and Benjamin Moll.
\newblock Present bias amplifies the household balance-sheet channels of macroeconomic policy.
\newblock \emph{The Quarterly Journal of Economics}, 140\penalty0 (1):\penalty0 691--743, 2025.

\bibitem[Moll(2014)]{moll2014productivity}
Benjamin Moll.
\newblock Productivity losses from financial frictions: Can self-financing undo capital misallocation?
\newblock \emph{American Economic Review}, 104\penalty0 (10):\penalty0 3186--3221, 2014.

\bibitem[Moll and Ryzhik(2025)]{moll2025mean}
Benjamin Moll and Lenya Ryzhik.
\newblock Mean field games without rational expectations.
\newblock \emph{arXiv preprint arXiv:2506.11838}, 2025.

\bibitem[Nose(2023)]{nose2023inclusive}
Manabu Nose.
\newblock \emph{Inclusive GovTech: Enhancing efficiency and equity through public service digitalization}.
\newblock International Monetary Fund, 2023.

\bibitem[P{\'a}stor and Veronesi(2020)]{pastor2020political}
Lubo{\v{s}} P{\'a}stor and Pietro Veronesi.
\newblock Political cycles and stock returns.
\newblock \emph{Journal of Political Economy}, 128\penalty0 (11):\penalty0 4011--4045, 2020.

\bibitem[Romer(1986)]{romer1986increasing}
Paul~M Romer.
\newblock Increasing returns and long-run growth.
\newblock \emph{Journal of Political Economy}, 94\penalty0 (5):\penalty0 1002--1037, 1986.

\bibitem[Romer(1990)]{romer1990endogenous}
Paul~M Romer.
\newblock Endogenous technological change.
\newblock \emph{Journal of Political Economy}, 98\penalty0 (5, Part 2):\penalty0 S71--S102, 1990.

\bibitem[Schaefer et~al.(2014)Schaefer, Schiess, and Wehrli]{schaefer2014long}
Andreas Schaefer, Daniel Schiess, and Roger Wehrli.
\newblock Long-term growth driven by a sequence of general purpose technologies.
\newblock \emph{Economic Modelling}, 37:\penalty0 23--31, 2014.

\bibitem[Solow(1957)]{solow1957technical}
Robert~M Solow.
\newblock Technical change and the aggregate production function.
\newblock \emph{The Review of Economics and Statistics}, 39\penalty0 (3):\penalty0 312--320, 1957.

\bibitem[Swan(1956)]{swan1956economic}
Trevor~Winchester Swan.
\newblock Economic growth and capital accumulation.
\newblock \emph{Economic Record}, 32\penalty0 (2):\penalty0 334--361, 1956.

\bibitem[Weidmann et~al.(2025)Weidmann, Xu, and Deming]{weidmann2025measuring}
Ben Weidmann, Yixian Xu, and David~J Deming.
\newblock Measuring human leadership skills with artificially intelligent agents.
\newblock Working Paper 33662, NBER, 2025.

\bibitem[Zhang et~al.(2018)Zhang, Yang, Chen, and Li]{zhang2018survey}
Qingchen Zhang, Laurence~T Yang, Zhikui Chen, and Peng Li.
\newblock A survey on deep learning for big data.
\newblock \emph{Information Fusion}, 42:\penalty0 146--157, 2018.

\end{thebibliography}

\end{document}